\documentclass[a4paper]{article}
\usepackage[utf8]{inputenc}
\usepackage{amsmath,amssymb,amsthm}
\usepackage{todonotes}
\usepackage{authblk}

\newtheorem{theorem}{Theorem}

\newtheorem{observation}{Observation}
\newtheorem{proposition}{Proposition}

\title{No Efficient Disjunction or Conjunction of Switch-Lists}
\author{Stefan Mengel}
\affil{Univ.~Artois, CNRS, Centre de Recherche en Informatique de Lens (CRIL)}

\begin{document}

\maketitle

\begin{abstract}
It is shown that disjunction of two switch-lists can blow up the representation size exponentially. Since switch-lists can be negated without any increase in size, this shows that conjunction of switch-lists also leads to an exponential blow-up in general.
\end{abstract}

\section{Introduction}

Switch-lists are a representation language for Boolean functions introduced in~\cite{CepekH17}, strongly related to interval representations~\cite{SchieberGZ05}.
The idea is to write the values of a Boolean function~$f$ on all lexicographically ordered inputs in a value
table. Then, to encode $f$, it suffices to remember the value of $f$ on the first input and the inputs at which the
value   of   $f$   changes   from   that   of   its   predecessor.   The   resulting   data   structure   is   called   a   \emph{switch-list}
representation of $f$. Clearly switch list representations can
be far more succinct than truth tables, e.g. for constant functions. 

To systematically understand the properties of switch-lists beyond this, Chrom\'{y} and \v{C}epek~\cite{CepekC20} analyzed them in the context of the so-called \emph{knowledge compilation map}. This framework, introduced in the ground-breaking work of Darwiche and Marquis~\cite{DarwicheM02} gives a list of standard properties which should be analyzed for languages used in the area of knowledge compilation along different axes: succinctness, queries and transformations. The idea of the knowledge compilation map has had a huge influence and the approach of~\cite{DarwicheM02} is widely applied in knowledge compilation, see e.g.~\cite{PipatsrisawatD08,FargierM14,FargierMN13} for a very small sample.

Chrom\'{y} and \v{C}epek~\cite{CepekC20} analyzed switch-lists along the properties of the knowledge compilation map and got a nearly complete picture. It turns out that switch-lists, while being generally much more succinct than truth tables, have many of their good properties. In particular, all of the queries in~\cite{DarwicheM02} (e.g.~consistency, entailement and counting) can be answered in polynomial time on switch-lists and nearly all of the  transformation can be performed efficiently. The only exception is that~\cite{CepekC20} leaves open if switch-lists are closed under bounded disjunction and bounded conjunction, i.e., given two Boolean functions $f_1$ and $f_2$ represented by switch-lists, can one compute a switch-list representation of $f_1\lor f_2$, resp.~$f_1\land f_2$, in polynomial time. It is shown here that this is not the case: there are Boolean functions $f_1$, $f_2$ such that any switch list representation of $f_1\lor f_2$ is exponentially larger than those of $f_1$ and $f_2$. This completes the analysis of switch-lists along the criteria of the knowledge compilation map and shows that (bounded) disjunction and conjunction are the only ``bad'' transformations of switch-lists, as there is no hope for a polynomial-time procedure in this case.

\section{Preliminaries}

Let $f$ be a Boolean function in the $n$ variables $\{x_1, \ldots, x_n\}$. Fix an order $\pi$ of $\{1, \ldots, n\}$. Then, the assignment $a:\{x_1,\ldots, x_n\}\rightarrow \{0,1\}$ can be identified with the number $b(a) \in \{0, \ldots , 2^n-1\}$ by identifying $a$ with $b(a) := \sum_{i=1}^n a(x_{\pi(i)}) 2^{i-1}$. This allows to write $a \prec a'$ if and only if $b(a)< b(a')$. The intuition behind all this is that the assignments are written in lexicographical order with respect to~$\pi$ and then $a \prec a'$ if and only if $a$ appears before $a'$.

A \emph{switch} of the function $f$ with respect to $\pi$ is a number $b\in \{1, \ldots , 2^n-1\}$ such that $f(b) \ne f(b-1)$ (note that here the identification of numbers and assignments to $\{x_1, \ldots, x_n\}$ depending on the order $\pi$ is used).
The \emph{switch-list} representation of $f$ with respect to $\pi$ consist of the value $f(0)$ and an ordered list of all switches of $f$ with respect to $\pi$.
Note that, for fixed $\pi$ the switch-list representation uniquely determines~$f$ and $f$ uniquely determines the switch-list representation. 

The \emph{size} of a switch-list representation is defined as $n$ times the number of switches which corresponds roughly to the natural encoding size\footnote{We do not take into account the size of an encoding of $\pi$ in this since it is the same for all switchlists in~$n$ variables and thus would only complicate the notion without giving any insights.}. Note that the size depends strongly on the order $\pi$.

Following Darwiche and Marquis~\cite{DarwicheM02}, switch-lists are said to satisfy bounded disjunction (resp.~bounded conjunction) if there is a polynomial-time algorithm that, given two switch-list representations of functions $f_1, f_2$, computes a switch-list representation of $f_1\lor f_2$ (resp.~$f_1\land f_2$). Chrom\'{y} and \v{C}epek~\cite{CepekC20} also considered the restricted version of bounded disjunction (resp.~conjunction) in which one assumes that the involved functions $f_1, f_2$ depend on the same set of variables.

\section{The Proof}

Let $n\in \mathbb{N}$ be even. Consider the functions $f_1(x_1, \ldots, x_n) := \Big(\bigwedge_{i=1}^{n/2} x_i\Big)\lor \Big(\bigwedge_{i=1}^{n} \neg x_i\Big)$ and $f_2(x_1, \ldots, x_n) := \Big(\bigwedge_{i=n/2+1}^{n} x_i\Big)\lor \Big(\bigwedge_{i=1}^{n} \neg x_i\Big)$.

\begin{observation}\label{obs}
There are switch-list representations for $f_1$ and $f_2$ with at most two switches.
\end{observation}
\begin{proof}
 Only give the argument for $f_1$ is given as that for $f_2$ is completely analogous. Fix any order $\pi$ in which the variables $x_1, \ldots, x_{n/2}$ come before those in $x_{n/2+1}, \ldots, x_n$. An assignment is a model of $f_1$ if and only if it maps all variables to $0$ or it maps $x_1, \ldots, x_{n/2}$ to $1$. So all models different from $0$ lie in the interval $[\sum_{j= n/2+1}^{n} 2^{j-1}, \sum_{j= 1}^{n} 2^{j-1}]$. Note that this interval lies at the end of the order of all assignments. So for these models, $f_1$ only has one switch at the beginning of the interval. To represent $f_1$ with a switch-list one only needs one additional switch directly after $0$.
\end{proof}

\begin{proposition}\label{prop}
 The function $f_1\lor f_2$ needs at least $2^{n/2+1}-3$ switches in any switch-list representation.
\end{proposition}
\begin{proof} 
 Let $X_1 := \{x_1, \ldots, x_{n/2}\}$ and $X_2:=\{x_{n/2+1} , \ldots, x_n\}$. Fix any variable order $\pi$ of $X_1\cup X_2$ and let $\preceq$ denote the lexicographical order with respect to~$\pi$. The last variable of $\pi$ is either in $X_1$ or in $X_2$. Without loss of generality,~assume that it is in $X_2$ and that the last variable in $\pi$ is $x_n$. 
 
 For every assignment $a$ to $X_1$, two extensions $e_0(a)$ and $e_1(a)$ to $X_1\cup X_2$ are constructed as follows: on $X_1$, the assignments $e_0(a)$ and $e_1(a)$ are both identical to $a$; all variables in $X_2 \setminus \{x_n\}$ are assigned $1$ and $x_n$ is assigned $0$ in $e_0(a)$ and~$1$ in $e_1(a)$. Let $\pi_1$ be the order $\pi$ restricted to $X_1$ and let $\preceq_1$ be the order of the assignments to $X_1$ with respect to $\pi_1$. Then for two assignments $e_i(a_1)$ and $e_j(a_2)$ it holds that $e_i(a_1) \prec e_j(a_2)$ if and only if $a_1 \prec_1 a_2$; or $a_1 =a_2$ and $i < j$. Note that none of the assignments of the form $e_i(a)$ is the constant $0$-assignment, so $e_i(a)$ satisfies $f_1\lor f_2$ if and only if it satisfies $\left(\bigwedge_{i=1}^{n/2} x_i\right)\lor \left(\bigwedge_{i=n/2+1}^{n} x_i\right)$.
 
 Now let $a_1, \ldots, a_{2^{n/2}-1}$ be the assignments to $X_1$ different from constant $1$-assignment given in the order $\preceq_1$. Then the resulting sequence 
 \begin{align}\label{sequence}e_0(a_1), e_1(a_1),\ldots, e_0(a_{2^{n/2}-1}), e_1(a_{2^{n/2}-1})\end{align} is in lexicographical order as well. Note that because none of the $a_i$ is the constant $1$-assignment, it holds that for every $i\in [2^{n/2}-1]$ that $e_1(a_i)$ is a model of $f_1\lor f_2$ while $e_0(a_i)$ is not. Thus there must be switches between each pair of consecutive elements of the sequence (\ref{sequence}). So there must be at least $2\times (2^{n/2}-1)-1 = 2^{n/2+1} -3$ switches in the switch-list representation of $f_1\lor f_2$ with respect to the order~$\pi$.
\end{proof}

The main result of this paper follows directly.
\begin{theorem}
 Switch-lists satisfy neither bounded disjunction nor bounded conjunction. This remains true when the functions to be disjoined (resp.~conjoined) are on the same set of variables.
\end{theorem}
\begin{proof}
For disjunction, this follows directly from Observation~\ref{obs} and Proposition~\ref{prop} since the outcome of any polynomial-time algorithm would in particular be of polynomial size.

For conjunction, let us define $f'_1 = \neg f_1$ and $f'_2 = \neg f_2$. Observe that a switch-list of $f$ can be negated in constant time by simply flipping the value $f(0)$ (keeping the same permutation of variables). Clearly $f'_1 \wedge f'_2 = \neg f_1 \wedge \neg f_2 = \neg (f_1 \vee f_2)$ and the lower bound for $f_1 \vee f_2$ from Proposition~\ref{prop} is of course valid also for $\neg (f_1 \vee f_2)$. This gives us an identical lower bound for the size of any switch list representing $f'_1 \wedge f'_2$.
\end{proof}

\section{Conclusion}

I was shown that switch-lists neither satisfy bounded disjunction nor bounded conjunction. This even remains true if both inputs depend on the same set of variables. This completes the analysis of switch-lists in the framework of the knowledge compilation map.

Let us remark that for practical applicability of switch-lists, this is rather bad news. Many classical approaches to practical knowledge compilation use so-called bottom-up compilation: given a conjunction of clauses, or more generally constraints, $F=\bigwedge_{i=1}^m C_i$, one first computes representations $R(C_i)$ of individual constraints $C_i$. Then one uses efficient conjunction to iteratively conjoin the $R(C_i)$  to get a representation of $F$. Since conjunction of even two switch-lists is hard in general, this approach is ruled out by our results.

To better understand when switch-lists are useful, it would be interesting to find classes of functions for which small switch-list representations can be computed efficiently, either theoretically or with heuristic approaches.

\paragraph*{Acknowledgement.} The author would like to thank Ondřej Čepek for helpful comments on an earlier version of this paper.
\bibliographystyle{plain}
\bibliography{switchlist}

\end{document}